\documentclass[12pt]{article}

\usepackage{latexsym,graphics,epsfig,amsmath,amssymb,tabularx,booktabs,color}

\usepackage{natbib}
\usepackage{enumerate}
\usepackage{amsthm}
\usepackage{url}

\theoremstyle{definition}
\newtheorem{example}{Example}[section]
\newtheorem{definition}{Definition}[section]

\theoremstyle{plain}
\newtheorem{theorem}{Theorem}[section]
\newtheorem{lemma}[theorem]{Lemma}

\newenvironment{keywords}{%
\begin{changemargin}{1cm}{1cm}
\noindent{\bf Keywords:}
}
{\end{changemargin} }

\newenvironment{changemargin}[2]{%
\begin{list}{}{%
\setlength{\topsep}{0pt}%
\setlength{\leftmargin}{#1}%
\setlength{\rightmargin}{#2}%
\setlength{\listparindent}{\parindent}%
\setlength{\itemindent}{\parindent}%
\setlength{\parsep}{\parskip}%
}%
\item[]}{\end{list}}

\newcommand*\xbar[1]{%
  \hbox{%
    \vbox{%
      \hrule height 0.5pt 
      \kern0.5ex
      \hbox{%
        \kern-0.25em
        \ensuremath{#1}%
        \kern-0.1em
      }%
    }%
  }%
}

\newcommand*\xtilde[1]{%
  \hbox{%
    \vbox{%
      \hrule height 0.5pt 
      \kern0.5ex
      \hbox{%
        \kern-0.25em
        \ensuremath{#1}%
        \kern-0.1em
      }%
    }%
  }%
}

\begin{document}

\title{On the closure of relational models}

\author{Anna Klimova \\ 
{\small{Institute of Science and Technology (IST) Austria}} \\
{\small \texttt{aklimova25@gmail.com}}\\
{}\\
\and Tam\'{a}s Rudas \\
{\small{E\"{o}tv\"{o}s Lor\'{a}nd University, Budapest, Hungary}} \\
{\small{\texttt{rudas@tarki.hu}}}}

\date{}

\maketitle

\begin{abstract}
\noindent Relational models for contingency tables are generalizations of log-linear
models, allowing effects associated with arbitrary subsets of cells in a possibly incomplete  table, and not necessarily containing the overall effect. In this generality, the MLEs under Poisson and multinomial sampling are not always identical. This paper deals with the theory of maximum likelihood estimation in the case when there are observed zeros in the data. A unique MLE to such data is shown to always exist in the set of pointwise limits of sequences of distributions in the original model. This set is equal to the closure of the original model with respect to the Bregman information divergence. The same variant of iterative scaling may be used to compute the MLE in the original model and in its closure. 
\end{abstract}


\section{Introduction}

The existence of maximum likelihood estimates under log-linear models for contingency tables has been thoroughly studied, see \cite{Haberman}, \cite{Andersen74}, \cite{Barndorff1978}, \cite{Lauritzen}, among others.  It was established that the maximum likelihood estimates of the cell parameters always exist if the observed table has only positive cell counts, and may exist if some of the observed counts are zero. The patterns of zero cells that lead to the non-existence of the MLE were described in several forms \citep[cf.][]{Haberman, FienbergRinaldo2012}.

Within the extended log-linear model class all data sets have an MLE, irrespective of the pattern of zeros. An extended log-linear model may be obtained as the closure of the original model in the topology of pointwise convergence \citep[cf.][]{Lauritzen}, or the closure with respect to the Kullback-Leibler divergence \citep[cf.][]{CsiszarMatus2003}, or as the aggregate exponential family \citep{BrownBook}.

The contribution of this paper is motivated by statistical problems in which models more general than log-linear need to be considered. To illustrate, suppose that the management of a large supermarket classifies all goods on stock into one of three mutually exclusive and exhaustive categories, say, food ($F$), non-food household ($N$) and other ($O$), and wishes to study how the daily sales of each group are related. This is a standard task in market basket analysis \citep[cf.][]{Brin1997}. The first model of interest, routinely, is independence, but the usual model of independence of the three indicator variables is not applicable in this case: if $p_F$, $p_N$ and $p_O$ denote the probabilities that a purchase (a basket) contains an item from the $F$, $N$ and $O$ groups, then the probability of an empty purchase would be $(1-p_F)(1-p_N)(1-p_O)$, which has to be positive, in spite of the fact that there are no purchases which do not contain any items.

One alternative independence concept to apply is the AS-independence of the three variables \citep{AitchSilvey60}. The indicator variables $F$, $N$, and $O$ are said to be AS-independent if
\begin{equation}\label{ASdef}
p_{FN} = p_Fp_N, \quad p_{FO} = p_Fp_O, \quad p_{NO} = p_Np_O, \quad p_{FNO} = p_Fp_Np_O.
\end{equation}
Relational models introduced by \cite*{KRD11} contain model (\ref{ASdef}) and many other models of association.

A relational model  on a contingency table is generated by a class of non-empty subsets of cells and can be specified in the form: 
\begin{equation}\label{genll}
\log \boldsymbol{\delta} = \textbf{A}' \boldsymbol{\beta}.
\end{equation}
Here, $\boldsymbol \delta$ denotes the vector of cell parameters, probabilities or intensities, and $\mathbf{A}$ is the 0-1 matrix whose rows are the indicators of generating subsets. A hierarchical log-linear model \citep*[cf.][]{BFH} applies to a table which is a Cartesian product, and the model is generated by a collection of cylinder sets corresponding to marginals of the table and thus is a special case of a relational model. If the row space of $\mathbf{A}$ contains the vector $\boldsymbol 1' = (1, \dots, 1)$, as in the case of hierarchical log-linear models, then the model is said to include the overall effect. A model with the overall effect can be parameterized to include a common parameter in every cell, often called the normalizing constant. The models without the overall effect cannot be parameterized in such a way.  The peculiar property of relational models without the overall effect is that models for probabilities (appropriate under multinomial sampling) and models for intensities (appropriate for Poisson sampling) are different and lead to different MLEs. Let $\boldsymbol y$ denote the observed frequency distribution. Then, when the overall effect is not present, the MLE for probabilities does not preserve the sufficient statistics $\mathbf{A} \boldsymbol y$, and, for intensities, it does not preserve the observed total $\boldsymbol 1'\boldsymbol y$, see Example \ref{AS3ex}.

An iterative scaling procedure based on Bregman divergence can be used to compute the MLE under relational models \citep{KRipf1}. The Bregman divergence between two distributions is a generalization of the Kullback-Leibler divergence, but, unlike the latter, stays non-negative whether or not the two distributions have the same total. This property is essential for relational models for intensities without the overall effect as these models may include distributions with different totals.

If the observed frequencies are positive and the model matrix is of full row rank, the MLE under relational models can be computed using algorithms  for convex optimization \citep[cf.][]{BertsekasNLP, AitchSilvey60, EvansForcina11} or the Newton-Raphson algorithm.  A detailed discussion of the relative advantages and disadvantages of variants of iterative proportional fitting was given in \cite{KRipf1}. The contribution of the present paper is the investigation of cases when there are observed zero frequencies in the data, and of the closure of relational models under which such data will always admit an MLE. Of course, if only three groups of goods, as in the example above, are investigated, one cannot expect to see an observed zero, but if $1000$ groups of goods are investigated, out of the resulting $2^{1000} - 1$  groups, many will be empty. As it turns out, the pattern of observed zeros has far reaching implications on the existence and kind of MLE obtained.

A necessary and sufficient condition for the existence of the maximum likelihood estimates of the cell parameters under relational models is obtained in Section \ref{RMmlePositive}. The MLE for $\boldsymbol y$ exists if and only if there is a positive vector $\boldsymbol z$ such that $\mathbf{A} \boldsymbol z = \mathbf{A} \boldsymbol y$. This is literally the same condition as the one that applies to log-linear models. 

In Section \ref{SecClosure}, extended relational models are studied. The extended relational model is defined as the set of distributions parameterized by the elements of an algebraic variety associated with the model matrix of the original relational model. It is shown that this set is equal to the closure of the original model with respect to both the pointwise convergence and the Bregman divergence.

In Section \ref{MLEsection}, a polyhedral condition for the existence of the MLE in the original or the extended relational model is formulated.  If the vector of the sufficient statistics, $\mathbf{A} \boldsymbol y$, of the observed distribution is not contained in any of the faces of the polyhedral cone associated with the model matrix, the MLE exists in the original model, and otherwise, it does in the extended model. This condition is the same as for the log-linear case, but the proof is very different. The multiplicative representation of the distributions in the extended model and the existence of the MLEs of the model parameters are also discussed in this section.  Finally, the generalized iterative proportional fitting procedure suggested in \cite{KRipf1} is extended to the case of observed zeros. 

While the conditions of the existence of the MLE in the generality considered in this paper may be formulated to coincide with the known conditions for the case of log-linear models,  the proofs turn out to be more involved. Also, the algorithm to obtain that the MLEs is more complex. The additional complications come from properties of the MLE when the overall effect is not present. In fact, \citet[p.75]{Lauritzen}  mentioned the existence of models without the overall effect, which he called the ``constant function'', but to avoid difficulties did not consider them.  On the other hand, such models have been used in practice, see references in \cite{KRD11, KRipf1}.

\nocite{CsiszarMatus2003}
\nocite{CsiszarMatus2008}

\section{MLE under relational models} \label{RMmlePositive}

Let $Y_1, \dots, Y_K$ be discrete random variables with finite ranges, and the vector $\mathcal{I}$ of length $|\mathcal{I}|$ be their joint sample space. Here, $\mathcal{I}$ may also be a proper subset of the Cartesian product of the ranges of the variables. A distribution on $\mathcal{I}$ is parameterized by the cell parameters $\boldsymbol \delta =\{\delta_i,\,\,\mbox{for } i \in \mathcal{I}\}$, and, to simplify notation, is identified with $\boldsymbol \delta$. The components of   $\boldsymbol \delta$ are either  probabilities: $\delta_i \equiv p_i \in (0,1)$, with $\sum_{i \in \mathcal{I}} p_i = 1$,  or intensities: $\delta_i \equiv \lambda_i > 0$, for all $i \in \mathcal{I}$. Let $\mathcal{P}$ denote the set of positive distributions, $\boldsymbol \delta > \boldsymbol 0$, on $\mathcal{I}$. 

Let $\mathbf{A}$ be a 0-1 matrix of size $J \times |\mathcal{I}|$, which is interpreted as the indicator matrix of $J$ subsets generating the model. Assume that $\mathbf{A}$ has no zero column. A relational model $RM_{\boldsymbol \delta}(\mathbf{A})$ is the following set of distributions:
\begin{equation} \label{RMexpF}
RM_{\boldsymbol \delta}(\mathbf{A}) = \{ \boldsymbol \delta \in \mathcal{P}: \,\,\delta_i = \prod_{j=1}^J \theta_j^{a_{ji}}, \, i \in \mathcal{I}, \, \mbox{ for some } \, \boldsymbol \theta \in \mathbb{R}^J_{>0}\},
\end{equation}
where $\boldsymbol \theta =(\theta_1, \dots, \theta_J) \in \mathbb{R}_{>0}^J$ denotes the vector of parameters associated with the generating subsets. Under the model, the cell parameters are equal to the products of the parameters $\boldsymbol \theta$ corresponding to the subsets to which the cell belongs. In the sequel, the components of $\boldsymbol \theta$ are referred to as the multiplicative parameters, and  $\mathbf{A}$ is assumed to be of full row rank. In fact, the model $RM_{\boldsymbol \delta}(\mathbf{A})$ is uniquely determined by the row space of its model matrix, $\mathcal{R}(\mathbf{A})$. Relational models for which $\boldsymbol 1' \in \mathcal{R}(\mathbf{A})$ are said to include the overall effect. 

A dual representation of a relational model $RM_{\boldsymbol \delta}(\mathbf{A})$ can be obtained using the kernel basis matrix $\mathbf{D}$, whose rows, $\boldsymbol d_1, \dots, \boldsymbol d_K$,  are a basis of $Ker(\mathbf{A})$. In this representation, any distribution in the model  satisfies
\begin{equation} \label{Dual}
\mathbf{D} \mbox{log } \boldsymbol \delta = \boldsymbol 0,
\end{equation}
which can be re-written using the generalized odds ratios:
\begin{equation}\label{pRatios}
\boldsymbol \delta^{\boldsymbol d_1^+} / \boldsymbol \delta^{\boldsymbol d_1^-} = 1, \quad 
\boldsymbol \delta^{\boldsymbol d_2^+} / \boldsymbol \delta^{\boldsymbol d_2^-} = 1, \quad
\cdots \quad
\boldsymbol \delta^{\boldsymbol d_K^+} / \boldsymbol \delta^{\boldsymbol d_K^-} = 1,
\end{equation}
or using the cross-product differences:
\begin{equation}\label{pDiff}
\boldsymbol \delta^{\boldsymbol d_1^+} - \boldsymbol \delta^{\boldsymbol d_1^-} = 0, \quad 
\boldsymbol \delta^{\boldsymbol d_2^+} - \boldsymbol \delta^{\boldsymbol d_2^-} = 0, \quad
\cdots \quad
\boldsymbol \delta^{\boldsymbol d_K^+} - \boldsymbol \delta^{\boldsymbol d_K^-} = 0,
\end{equation}
where, ${\boldsymbol {d^+}}$ and ${\boldsymbol {d^-}}$ denote, respectively, the positive and negative parts of a vector $\boldsymbol d$ \citep{KRD11}.  

The properties of the maximum likelihood estimates under relational models are reviewed next. Let $\mathbf{Y} = (Y_1, \dots, Y_K)$ be a random variable that has a multivariate Poisson distribution parameterized by $\boldsymbol \delta \equiv \boldsymbol \lambda$ or a multinomial distribution parameterized by $N$ and $\boldsymbol \delta \equiv \boldsymbol p$. Let $\boldsymbol y$ be a realization of $\mathbf{Y}$, and \begin{equation}\label{b}
\boldsymbol q = \left\{\begin{array}{ll} \boldsymbol y, & \mbox{if } \boldsymbol \delta \equiv \boldsymbol \lambda, \\
\boldsymbol y/ (\boldsymbol 1'\boldsymbol y), & \mbox{if } \boldsymbol \delta \equiv \boldsymbol p.
\end{array}\right.
\end{equation}

If the MLE $\hat{\boldsymbol \delta}_{\boldsymbol y}$ of the cell parameters under the model $RM_{\boldsymbol \delta}(\mathbf{A})$ exists, it is the unique solution to the system of equations:
\begin{align}\label{MLEsys}
&\mathbf{A}\boldsymbol{\delta} = \gamma \mathbf{A} \boldsymbol q, \nonumber \\
&\mathbf{D} \mbox{log } \boldsymbol{\delta} = \boldsymbol 0,\\
&\boldsymbol 1'\boldsymbol \delta = 1 \qquad (\mbox{only for } \boldsymbol \delta \equiv \boldsymbol p). \nonumber
\end{align}
The value of $\gamma$ is called the adjustment factor. If $RM_{\boldsymbol \delta}(\mathbf{A})$ is a model for probabilities with the overall effect or a model for intensities, then $\gamma = 1$ for every $\boldsymbol y$. If $RM_{\boldsymbol \delta}(\mathbf{A})$ is a model for probabilities without the overall effect, then the value of $\gamma$ depends on $\boldsymbol y$  \citep{KRD11}.

\begin{example}\label{AS3ex}
The model of AS-independence (\ref{ASdef}) is a relational model generated by the model matrix
\begin{equation}\label{ASmodelMx}
\mathbf{A} = \left(\begin{array}{ccccccc} 
1&0&0& 1 & 1& 0& 1\\
0&1&0& 1 & 0 & 1& 1\\
0&0&1& 0& 1& 1 & 1
\end{array} \right),
\end{equation}
where the order of cells is lexicographic. As $\boldsymbol 1'$ is not in the row space of $\mathbf{A}$, the model does not have the overall effect. Thus, the models $RM_{\boldsymbol \lambda} (\mathbf{A})$ and $RM_{\boldsymbol p} (\mathbf{A})$ are not equivalent. Given hypothetical data, the MLE for cell frequencies,  computed under the model for probabilities and under the model for intensities, are shown in Table \ref{TableMLEbothAS}. In the case of probabilities, the estimates for sufficient statistics are about $0.7$ times less than the observed sufficient statistics. In the case of intensities, the estimated total is approximately  $79.73$, while the observed total is $100$. The estimates were obtained using the R-package {\tt{gIPFrm}} \citep{gIPFpackage}. 
\begin{table}
\centering
\caption{Maximum likelihood estimates under the model of AS-independence of variables $F$, $N$, $O$, under the multinomial and Poisson sampling. }
\vspace{5mm}
\begin{tabular}{lccccl}
\cline{1-5} \\[-5pt]
   & \multicolumn{2}{c}{$O = \mbox{\it{No}}$ }&\multicolumn{2}{c}{$O = \mbox{\it{Yes}}$}&\\
\cline{1-5} \\[-2pt]
   & $N=\mbox{\it{No}}$ & $N = \mbox{\it{Yes}}$ & $N = \mbox{\it{No}}$ & $N=\mbox{\it{Yes}}$ &\\ [1pt]
\cline{1-5} \\ 
$F = \mbox{\it{No}}$ & {\it empty}  & \multicolumn{1}{c}{$14$ } & \multicolumn{1}{c}{$25$ } & \multicolumn{1}{c}{$16$ }& {\it - observed} \\ [4pt]
&{\it empty}  &$27.33$ & $32.60$ & $8.91$ & {\it - multinomial}\\
& {\it empty}  &$3.31$ &$7.29$ & $24.13$ &{\it - Poisson}\\ [8pt]
$F = \mbox{\it{Yes}}$ & \multicolumn{1}{c}{$10$ }  & \multicolumn{1}{c}{$5$} & \multicolumn{1}{c}{$3$} & \multicolumn{1}{c}{$27$} &{\it - observed} \\ [4pt]
&  $18.46$  & $ 5.04$ & $6.02$ & $1.64$ &{\it - multinomial}\\
&  $1.26$  & $4.17$ & $9.18$ & $30.39$ &{\it - Poisson}\\
\end{tabular} 
\label{TableMLEbothAS}
\end{table}
\qed
\end{example}

A necessary and sufficient condition for the existence of the MLE is given in the next theorem. Its proof uses the following lemma: 

\begin{lemma}\label{PositiveExists}
If $\boldsymbol y > \boldsymbol 0$, the MLE $\hat{\boldsymbol \delta}_{\boldsymbol y}$ exists.
\end{lemma}

\begin{proof}

A relational model for intensities is a regular exponential family \citep{KRD11}, and the standard proof applies \citep[cf.][]{Andersen74}.

In the case of probabilities, $\boldsymbol \delta \equiv \boldsymbol p$, the MLE, if exists, is the unique solution to (\ref{MLEsys}). \citet[Lemma 3.5]{KRipf1} showed that there exist $\gamma_1, \gamma_2 > 0$ such that the adjustment factor $\gamma \in [\gamma_1, \gamma_2]$.  Since $\gamma \boldsymbol y > \boldsymbol 0$, the MLE  $\hat{\boldsymbol \lambda}_{\gamma \boldsymbol y}$ under the model for intensities  $RM_{\boldsymbol \lambda}(\mathbf{A})$ exists for every $\gamma \in [\gamma_1, \gamma_2]$, and, by Lemma 3.6 in \cite{KRipf1}, one can find a unique $\gamma^*$ such that $\boldsymbol 1'\hat{\boldsymbol \lambda}_{\gamma^*\boldsymbol y} = 1$. Because $\hat{\boldsymbol \lambda}_{\gamma^*\boldsymbol y}$ satisfies (\ref{MLEsys}),  $\hat{\boldsymbol p}_{\boldsymbol y} = \hat{\boldsymbol \lambda}_{\gamma^*\boldsymbol y}$. 
\end{proof}

As shown next, the MLE may exists when some of the observed frequencies are zero.

\vspace{4mm}

\textbf{Example \ref{AS3ex}} (revisited):

Let $\boldsymbol q = (0,0,0,0,0,0,1)'$ be the observed distribution. Under the model of AS-independence, the MLEs for cell probabilities exist and are equal to $$\hat{\boldsymbol p} = \left(\sqrt[3]{2} - 1, \sqrt[3]{2} - 1, \sqrt[3]{2} - 1, (\sqrt[3]{2} - 1)^2, (\sqrt[3]{2} - 1)^2, (\sqrt[3]{2} - 1)^2, (\sqrt[3]{2} - 1)^3\right)'.$$ \qed

\begin{theorem}\label{MLEnoOvEff}
Let $\boldsymbol y$ be the vector of observed frequencies under Poisson or multinomial sampling, and let $RM_{\boldsymbol \delta}(\mathbf{A})$ be a relational model. The MLE $\hat{\boldsymbol \delta}_{\boldsymbol y}$ under the model exists if and only if there is a positive vector $\boldsymbol z$, such that $\mathbf{A}\boldsymbol z = \mathbf{A}\boldsymbol q$, with $\boldsymbol q$ defined in (\ref{b}).
\end{theorem}

\begin{proof}

In the case of intensities, $\boldsymbol \delta \equiv \boldsymbol \lambda$, the standard proof for regular exponential families \citep[cf.][]{Andersen74} applies.

The case of probabilities, $\boldsymbol \delta \equiv \boldsymbol p$, is considered next.
Suppose $\hat{\boldsymbol p}_{\boldsymbol y} > 0$ exists. By Corollary 4.2 in \cite{KRD11}, $\,\, \mathbf A\hat{\boldsymbol p}_{\boldsymbol y} = \gamma  \mathbf A\boldsymbol q$ for some $\gamma > 0$. 
Therefore, $\hat{\boldsymbol p}_{\boldsymbol y} = \gamma \boldsymbol q + \boldsymbol d$,
for some $\boldsymbol d \in Ker (\mathbf A)$. Take  
$$\boldsymbol z = \frac{1}{\gamma}\hat{\boldsymbol p}_{\boldsymbol y} = \boldsymbol q +  \frac{1}{\boldsymbol \gamma}\boldsymbol d > \boldsymbol 0.$$
Then, as $\frac{1}{\boldsymbol \gamma}\boldsymbol d \in Ker(\mathbf{A})$, $\,\,\mathbf{A}\boldsymbol z = \mathbf{A} \boldsymbol q$, as required. 

To prove the converse, assume that there exists a $\boldsymbol z > \boldsymbol 0$, such that $\mathbf{A}\boldsymbol z = \mathbf{A} \boldsymbol q$. Thus,  $\boldsymbol z = \boldsymbol q + \boldsymbol d$ for some $\boldsymbol d \in Ker (\mathbf A)$. 
Let $$\boldsymbol d_1 =  \frac{1}{1+\boldsymbol 1'\boldsymbol d}\boldsymbol d,$$
and note that $1+\boldsymbol 1'\boldsymbol d = \boldsymbol 1'\boldsymbol q + \boldsymbol 1'\boldsymbol d = \boldsymbol 1' \boldsymbol z > 0$. 
Next, consider  $\boldsymbol v = (1-\boldsymbol 1'\boldsymbol d_1)\boldsymbol q + \boldsymbol d_1$. Then $\boldsymbol 1'\boldsymbol v = (1-\boldsymbol 1'\boldsymbol d_1) + \boldsymbol 1'\boldsymbol d_1 = 1$, and
$$
\boldsymbol v  =  (1-\boldsymbol 1'\boldsymbol d_1)\boldsymbol q + \boldsymbol d_1 = \frac{1}{1+\boldsymbol 1'\boldsymbol d}\boldsymbol q +  \frac{1}{1+\boldsymbol 1'\boldsymbol d}\boldsymbol d =  \frac{1}{1+\boldsymbol 1'\boldsymbol d}(\boldsymbol q + \boldsymbol d) > 0.
$$
Therefore, $\boldsymbol v$ is a positive probability distribution, and, by Lemma \ref{PositiveExists}, the MLE $\hat{\boldsymbol p}_{\boldsymbol v}$ exists, and satisfies:
\begin{align*}
&\mathbf{A}\hat{\boldsymbol p}_{\boldsymbol v} = \gamma_{\boldsymbol v} \mathbf{A} \boldsymbol v, \\
&\mathbf{D} \mbox{log } \hat{\boldsymbol p}_{\boldsymbol v} = \boldsymbol 0,\\
&\boldsymbol 1'\hat{\boldsymbol p}_{\boldsymbol v} = 1,
\end{align*}
for some $\gamma_{\boldsymbol v} > 0$.   Then, from the definition of $\boldsymbol v$, $\mathbf{A} \hat{\boldsymbol p}_{\boldsymbol v} = \gamma_{\boldsymbol v} \mathbf{A} \boldsymbol v = \gamma_{\boldsymbol v}  (1-\boldsymbol 1'\boldsymbol d_1)\mathbf{A} \boldsymbol q$, that is, $\boldsymbol p_{\boldsymbol v}$ is also the MLE for $\boldsymbol q$ with the adjustment factor $\gamma = \gamma_{\boldsymbol v}  (1-\boldsymbol 1'\boldsymbol d_1)$. 
\end{proof}

The statement of the theorem is illustrated in the next example. 

\begin{example}\label{exConstruct}

Let $RM_{\boldsymbol p}(\mathbf{A})$ be the model for probabilities generated by 
\begin{equation*}
\mathbf{A} = \left( 
\begin{array}{ccccc}
1&1&1&0&1\\
1&1&0&0&1\\
1&0&0&1&1
\end{array}
\right),
\end{equation*}
and let $\boldsymbol q = (3/7, 3/7, 0 , 1/7,0)'$ be the observed probability distribution. Consider any vector $\boldsymbol z$, whose subset sums, $\mathbf{A}\boldsymbol z$, are equal to the observed subset sums: 
\begin{align*}
&z_1 + z_2 + z_3 + z_5 = 6/7, \quad
z_1 + z_2 + z_5 = 6/7, \quad z_1 + z_4 + z_5 = 4/7.
\end{align*}
The first two equations imply that $z_3 = 0$. Therefore, there is no (strictly) positive distribution with the same subset sums as those observed, and thus, $\boldsymbol q$ does not have an MLE in the model.  \qed
\end{example}

In the next section, an extended relational model is defined as the polynomial variety corresponding to the model matrix. It is further shown that the extended model coincides with the set of pointwise limits of sequences of distributions in the original model and is also the closure with respect to Bregman information divergence.

\section{Extended relational models} \label{SecClosure}

Let $\mathbf{A}$ be the model matrix of a relational model, and let $\mathcal{X}_{\mathbf A}$ denote the polynomial  variety associated with $\mathbf{A}$ \citep{SturBook}:
\begin{equation}\label{variety}
\mathcal{X}_{\mathbf A} = \left\{\boldsymbol \delta \in \mathbb{R}^{|\mathcal{I}|}_{\geq 0}: \,\, \boldsymbol \delta^{\boldsymbol d^+} = \boldsymbol \delta^{\boldsymbol d^-}, \,\,  \forall  \boldsymbol d \in Ker(\mathbf{A}) \right \}.
\end{equation}
 
\begin{definition}\label{ExtendDef}
The extended relational model for intensities, $\xbar{RM}_{\boldsymbol \lambda} (\mathbf{A})$, is the set of distributions 
\begin{equation}\label{ExtMint}
\boldsymbol \lambda \in \mathcal{X}_{\mathbf A}.
\end{equation}
The extended relational model for probabilities, $\,\xbar{RM}_{\boldsymbol p} (\mathbf{A})$, is  the set of distributions 
\begin{equation}\label{ExtMprob}
\boldsymbol p \in \mathcal{X}_{\mathbf A} \cap \Delta_{|\mathcal{I}|-1},
\end{equation} 
where $\Delta_{|\mathcal{I}|-1}$ is the $(|\mathcal{I}|-1)$-dimensional simplex.\qed
\end{definition}

For positive distributions being in $\mathcal{X}_{\mathbf{A}}$ is equivalent to the representations (\ref{Dual}), (\ref{pRatios}), and (\ref{pDiff}). Therefore, the   relational model generated by $\mathbf{A}$ is a subset of the corresponding extended model. For a positive $\boldsymbol \delta$, whether or not  (\ref{Dual}), (\ref{pRatios}), and (\ref{pDiff}) hold does not depend on the choice of $\mathbf{D}$. However, as illustrated next, there exist  $\boldsymbol \delta \geq \boldsymbol 0$, which, due to the pattern of zeros, satisfy (\ref{pDiff}) for some choice of $\mathbf{D}$ and do not satisfy for another. 

\vspace{5mm}

\noindent{\bf Example \ref{AS3ex}} (revisited):
The model has dual representations using matrices $\mathbf{D}_1$ and $\mathbf{D}_2$:
$$\mathbf{D}_1 = \left(\begin{array} {rrrrrrr}
1&1&0&-1&0&0&0\\
1&0&1&0&-1&0&0\\ 
0&1&1&0&0&-1&0 \\
1&1&1&0&0&0&-1 \\
\end{array}\right) \,\,\,  \,\,\,
\mathbf{D}_2 = \left(\begin{array} {rrrrrrr}
0&0&1&1&0&0&-1\\
0&1&0&0&1&0&-1\\ 
1&0&0&0&0&1&-1\\
1&1&1&0&0&0&-1 \\
\end{array}\right).$$
The distribution $\boldsymbol \delta = (0,0,0,1,1,1,0)'$ satisfies (\ref{pDiff}) if obtained from $\mathbf{D}_2$, but does not satisfy (\ref{pDiff}) if obtained using $\mathbf{D}_1$, and therefore, $\boldsymbol \delta \notin \mathbf{X}_{\mathbf{A}}$. 
\qed
\vspace{2mm}

The support $supp(\boldsymbol \delta) = \{i \in \mathcal{I}: \,\, \delta_i > 0\}$ of   distributions with zero components which are in $\mathcal{X}_{\mathbf{A}}$ can be characterized using the concept of a facial set which is defined next.

Let $\boldsymbol a_1, \dots, \boldsymbol a_{|\mathcal{I}|}$ denote the columns of $\mathbf{A}$, and let $C_{\mathbf{A}}$  be the set of all non-negative linear combinations of these columns: 
\begin{equation}\label{IntCone}
C_{\mathbf{A}} = \{ \boldsymbol t \in \mathbb{R}^J_{\geq 0}: \,\, \exists \boldsymbol\delta \in \mathbb{R}^{|\mathcal{I}|}_{\geq 0}  \quad \boldsymbol t = \mathbf{A}\boldsymbol \delta\}. 
\end{equation}
The relative interior of $C_{\mathbf{A}}$, $relint(C_{\mathbf{A}})$, comprises such $\boldsymbol t \in \mathbb{R}^J_{> 0}$, for which there exists a (strictly) positive $\boldsymbol\delta$ that satisfies $\boldsymbol t = \mathbf{A}\boldsymbol \delta$.

The set $C_{\mathbf{A}}$ is a polyhedral cone in $\mathbb{R}^J$. If an affinely independent set $\boldsymbol a_{i_1}, \boldsymbol a_{i_2}, \dots, \boldsymbol a_{i_f}$ of columns of $\mathbf{A}$ spans a proper face of $C_{\mathbf{A}}$, the set of indices $F = \{i_1, i_2, \dots,i_f\}$ is called  facial \citep[cf.][]{GrunbaumConvex, GeigerMeekSturm2006}.  The facial sets of $\mathbf{A}$ are determined by its row space \citep[cf.][]{FienbergRinaldo2012}. If $\boldsymbol t \in C_{\mathbf{A}}\setminus relint(C_{\mathbf{A}})$, then $\boldsymbol t$ is said to lie on a face of $C_{\mathbf{A}}$. In that case, there is a facial set $F = F(\boldsymbol t)$, such that
\begin{equation}\label{linComb}
\boldsymbol t = s_1\boldsymbol a_{i_1} + \dots + s_f\boldsymbol a_{i_f}.
\end{equation}
Equivalently, a set $F$ is facial if and only if there exists a $\boldsymbol c \in \mathbb{R}^J$, such that $\boldsymbol c'\boldsymbol a_i = 0$ for every $i \in F$ and $\boldsymbol c'\boldsymbol a_i  > 0$ for every $i \notin F$. The properties of facial sets are formulated in Lemma \ref{LemmaFacial} given in the Appendix. In particular, only  distributions whose support is $\mathcal{I}$ or a facial set of $\mathbf{A}$ may belong to $\mathcal{X}_{\mathbf{A}}$. As an example, the facial sets of the model matrix (\ref{ASmodelMx}) of AS-independence are $\{1\}$, $\{2\}$, $\{3\}$, $\{1,2,4\}$, $\{2,3,6\}$, $\{1,3,5\}$. The support  $\{4,5,6\}$ of $\boldsymbol \delta =(0,0,0,1,1,1,0)'$ from Example \ref{AS3ex} is not a facial set,  and thus $\boldsymbol \delta$ cannot be an element of $\mathcal{X}_{\mathbf{A}}$.

The following theorem describes the structure of the parameter set of the extended relational model.

\begin{theorem}\label{THclosure}
The extended relational model $\,\xbar{RM}_{\delta}(\mathbf{A})$ is the closure of the relational model $RM_{\boldsymbol \delta}(\mathbf{A})$ in the topology of pointwise convergence:
$\,\xbar{RM}_{\delta}(\mathbf{A}) = cl (RM_{\boldsymbol \delta}(\mathbf{A}))$.
\end{theorem}

The proof is provided in the Appendix. The theorem says that every distribution in the extended model can be obtained as a pointwise limit of a sequence of distributions in the non-extended model.  In the following example, such a sequence is found using the construction described in the proof.  

\vspace{5mm}

\noindent{\bf Example \ref{AS3ex}} (revisited):

\vspace{2mm}

The set $F=\{2,3,6\}$ is facial set of $\mathbf{A}$, and thus, by Lemma \ref{LemmaFacial}, the extended model contains a distribution $\boldsymbol p = (0, p_2, p_3, 0 , 0, p_6, 0)'$, where $p_2, p_3, p_6 > 0$ and $p_2 + p_3 + p_6 = 1$. To construct a sequence of distributions in the original model which converges to $\boldsymbol p$, find $\theta_2, \theta_3$ such that  
$$\theta_2 =   p_2, \,\, \theta_3 = p_3, \,\, \theta_2\theta_3 = p_6.$$
From the normalization condition, $$\theta_2 = \frac{1 - \theta_3}{1+\theta_3}.$$
Take an arbitrary $\theta_1 \in (0,1)$, then set
$$\theta_2^{(n)} = \frac{1 - \theta_1n^{-1} - \theta_3 - \theta_1 n^{-1} \theta_3}{1+ \theta_1n^{-1} + \theta_3 + \theta_1 n^{-1} \theta_3},$$
and consider
$$\boldsymbol p^{(n)} = (\theta_1n^{-1}, \quad \theta_2^{(n)}, \quad \theta_3, \quad \theta_1n^{-1} \theta_2^{(n)}, \quad \theta_1n^{-1} \theta_3,  \quad \theta_2^{(n)}\theta_3, \quad \theta_1n^{-1}\theta_2^{(n)}\theta_3)'.$$
For every $n$, $\,\boldsymbol p^{(n)} \in RM_{\boldsymbol p}(\mathbf{A})$. 
As $n \to \infty$, $\theta_2^{(n)} \to \theta_2$, and therefore, $\boldsymbol p^{(n)} \to \boldsymbol p$. The construction is complete. \qed

\vspace{3mm}

An extended relational model can also be defined as a closure of the  exponential family corresponding to the original model. The closure of exponential families using the Kullback-Leibler divergence was described for regular families by  \cite{BrownBook}, among others, and for full families by \cite{CsiszarMatus2003}.  However, both of these approaches rely on the presence of the overall effect, which implies, through the possibility of normalization, that the Kullback-Leibler divergence is non-negative and Pinsker's inequality \citep[cf.][]{Csiszar} holds. In the generality considered in the present paper, the approach does not apply, and the Bregman divergence is used to define the closure.

Let $D(\cdot|| \cdot)$ denote the Bregman divergence between two vectors  $\boldsymbol t, \boldsymbol u \in \mathbb{R}^{|\mathcal{I}|}_{> 0}$, associated with the function $f(\boldsymbol x) = \sum_{i \in \mathcal{I}} x(i) \mbox{log }x(i)$: 
\begin{equation} \label{BDdef} 
D(\boldsymbol t||\boldsymbol u) = \sum_{i \in \mathcal{I}} t(i) \mbox{log }(t(i)/u(i)) + (\sum_{i \in \mathcal{I}} u(i) - \sum_{i \in \mathcal{I}} t(i)).
\end{equation}
Under the convention $0\cdot\log 0 = 0$, $D(\boldsymbol t||\boldsymbol u)$ is also defined for non-negative $\boldsymbol t$ and $\boldsymbol u$ if $supp(\boldsymbol t) \subseteq supp(\boldsymbol u)$. The function $D(\boldsymbol t|| \boldsymbol u)$ is non-negative, and $D(\boldsymbol t|| \boldsymbol u) = 0$ if and only if $\boldsymbol t = \boldsymbol u$. For any $\boldsymbol u^* \in \mathbb{R}^{|\mathcal{I}|}_{\geq 0}$ and for any convex set $\mathcal{S} \subset  \mathbb{R}^{|\mathcal{I}|}_{\geq 0}$ there exists a unique $\boldsymbol u^* \in \mathbb{R}^{|\mathcal{I}|}_{\geq 0}$, such that 
\begin{equation}\label{BregmanPr}
D(\boldsymbol u^*|| \boldsymbol u) = \min_{\boldsymbol z \in \mathcal{S}}  D(\boldsymbol z|| \boldsymbol u),
\end{equation}
see  \cite{Bregman}. This $\boldsymbol u^*$ is called the D-projection, or the Bregman projection, of $\boldsymbol u$ on $\mathcal{S}$.  If $\boldsymbol p_1$ and $\boldsymbol p_2$ are probability distributions,  then $D(\boldsymbol p_1|| \boldsymbol p_2)$ is the Kullback-Leibler divergence.

Let $\widetilde{RM}_{\delta}(\mathbf{A})$ be the closure of $RM_{\boldsymbol \delta}(\mathbf{A})$ with respect to the Bregman divergence:

$$\widetilde{RM}_{\delta}(\mathbf{A}) =\left\{\boldsymbol \delta \in \bar{\mathcal P}: \,\, \exists \boldsymbol \delta^{(n)} \in {RM}_{\delta}(\mathbf{A}), n \in \mathbb{N}, \mbox{ such that} \,\, D(\boldsymbol \delta || \boldsymbol \delta^{(n)}) \to 0 \, \mbox{ as } n \to \infty\right\}.$$

\begin{theorem}\label{ThClosureBregman}
The closures of the relational model $RM_{\boldsymbol \delta}(\mathbf{A})$ according to the pointwise convergence and to the Bregman divergence coincide. 
\end{theorem}

\begin{proof}

Let $\boldsymbol \delta^* \in \xbar{RM}_{\boldsymbol \delta}(\mathbf{A})$. 
Then, there exists a sequence $\boldsymbol \delta^{(n)} \in RM_{\boldsymbol \delta}(\mathbf{A})$ such that $\boldsymbol \delta^{(n)} \to \boldsymbol \delta^*$ pointwise, as $n \to \infty$. The function $D(\boldsymbol \delta^*|| \boldsymbol \delta^{(n)})$ is defined and continuous for $\boldsymbol \delta^{(n)} > 0$, even if some of the components of $\boldsymbol \delta^*$ are zero. Therefore, $D(\boldsymbol \delta^*|| \boldsymbol \delta^{(n)}) \to 0$, as $n \to \infty$. 

Suppose $\boldsymbol \delta^* \in \widetilde{RM}_{\boldsymbol \delta}(\mathbf{A})$, and, thus, there exists a sequence $\boldsymbol \delta^{(n)} \in RM_{\boldsymbol \delta}(\mathbf{A})$, such that:
$$D(\boldsymbol \delta^* || \boldsymbol \delta^{(n)}) \to 0 \quad \mbox{ as } n \to \infty.$$
Therefore, $D(\boldsymbol \delta^* || \boldsymbol \delta^{(n)}) \leq 1$ for all large enough $n$. Because the set $\{\boldsymbol \delta \geq \boldsymbol 0: \,\, D(\boldsymbol \delta^*|| \boldsymbol \delta) \leq 1\}$ is compact in $\mathbb{R}^{|\mathcal{I}|}$ \citep{Bregman}, there exists a subsequence $\boldsymbol \delta^{(n_k)}$ that converges pointwise to $\boldsymbol \delta^*$, as $k \to \infty$.
\end{proof}

A relational model $RM_{\boldsymbol \delta}(\mathbf{A})$ is a multiplicative family of distributions; the conditions under which the extended model $\, \xbar{RM}_{\boldsymbol \delta}(\mathbf{A})$ is also a multiplicative family are studied next.

A distribution $\boldsymbol \delta \in \bar{\mathcal{P}}$ is said to factor according to a matrix $\mathbf{A}$ if it has a representation given in (\ref{RMexpF}), with $\boldsymbol \theta = (\theta_1, \dots, \theta_J)' \geq \boldsymbol 0$. Every distribution in a relational model factors according to the model matrix. However, as the next example demonstrates, an extended model may contain distributions which do not factor according to one choice of the model matrix but do factor according to a different choice.

\vspace{5mm}

\noindent{\bf Example \ref{exConstruct}} (revisited):
Any distribution in $RM_{\boldsymbol p}(\mathbf{A})$ factors according to $\mathbf{A}$, that is,
\begin{equation}\label{Amult} 
\boldsymbol p =(\theta_1\theta_2\theta_3, \, \theta_1\theta_2, \, \theta_1, \, \theta_3, \, \theta_1\theta_2\theta_3)',
\end{equation}
for some $\theta_1, \theta_2, \theta_3 > 0$. The non-negative distribution $\boldsymbol p_0 = (1/8, 1/2, 0, 1/4, 1/8)'$ does not have the multiplicative structure (\ref{Amult}), but is in the extended model. To show the latter, take $$\theta^{(n)}_1 =  \frac{3}{3n+4}, \quad \theta^{(n)}_2 = \frac{n}{2}, \quad \theta^{(n)}_3 = \frac{1}{4}, \quad n \geq 1.$$ Then, the sequence
$$\boldsymbol p^{(n)} = \left( \frac{3n}{8(3n+4)}, \frac{3n}{2(3n+4)}, \frac{3}{3n+4}, \frac{1}{4}, \frac{3n}{8(3n+4)}\right)'$$
is in the model, and $\lim_{n \to \infty} \boldsymbol p^{(n)} =  \boldsymbol p_0$. 
On the other hand, $\boldsymbol p_0$ factors according to the matrix 
\begin{equation*}
\mathbf{A}_1 = \left( 
\begin{array}{ccccc}
0&0&1&0&0\\
1&1&0&0&1\\
1&0&0&1&1
\end{array}
\right),
\end{equation*}
which generates the same extended model as $\mathbf{A}$ does, because  $Ker(\mathbf{A}) = Ker(\mathbf{A}_1)$. 
\qed

\vspace{5mm}

A necessary and sufficient condition of the existence of such a factorization for a distribution in an extended relational model is given next.

\begin{theorem}\label{Thfactors}
A distribution $\boldsymbol \delta  \in \,\xbar{RM}_{\delta} (\mathbf{A})$  factors according to $\mathbf A$ if and only if 
for any $i_0 \notin supp(\boldsymbol \delta)$ there exists an index $j = j(i_0) \in \{1, \dots, J\}$ such that
$a_{ji} = 0$ for all $i \in supp(\boldsymbol \delta)$. \qed
\end{theorem}

The condition of the theorem, called the $\mathbf{A}$-feasibility of $supp(\boldsymbol \delta)$, means that a generating subset which contains a zero cell of the distribution does not include any positive cell.  
For extended log-linear models, this condition was proved in \cite{GeigerMeekSturm2006} and \cite{RauhMatroid}. The proofs given did not actually rely on the presence of the overall effect and thus apply here.

Maximum likelihood estimation in the extended relational model is studied next.

\section{MLE in the extended model}\label{MLEsection}

Let $F$ be a facial set, and let $\mathbf{A}_F$ denote the sub-matrix of $\mathbf{A}$ comprising the columns with indices in $F$, and $\boldsymbol \delta_F$ denote the sub-vector of $\boldsymbol \delta$ with indices in $F$. The following result extends Theorem  9 in \cite{FienbergRinaldo2012}.

\vspace{3mm}

\begin{theorem}\label{MLEextendTHnew}
Let $\boldsymbol y$ be the vector of observed frequencies under Poisson or multinomial sampling, and let $RM_{\boldsymbol \delta}(\mathbf{A})$ be a relational model.  Consider $\boldsymbol q$ defined in (\ref{b}), and assume that $supp(\boldsymbol q) \subsetneq \mathcal{I}$.
\begin{enumerate}[(i)]
\item If for all facial sets $F$, $supp(\boldsymbol q) \not\subseteq F$, then the MLE $\tilde{\boldsymbol \delta}_{\boldsymbol y}$ under the model $\,\xbar{RM}_{\boldsymbol \delta}(\mathbf{A})$ exists, and is also the MLE under  $RM_{\boldsymbol \delta}(\mathbf{A})$: $\tilde{\boldsymbol \delta}_{\boldsymbol y} = \hat{\boldsymbol \delta}_{\boldsymbol y}$. Otherwise, 
\item Let $F$ be the smallest facial set such that $supp(\boldsymbol q) \subseteq F$.  Then the MLE $\hat{\boldsymbol \delta}_{\boldsymbol y, F}$ of ${\boldsymbol \delta}_F$ under  the model ${RM}_{\boldsymbol \delta_F}(\mathbf{A}_F)$ exists, and  $\tilde{\boldsymbol \delta}_{\boldsymbol y} = (\hat{\boldsymbol \delta}_{\boldsymbol y, F}, \boldsymbol 0_{\mathcal{I}\setminus F})$ is the MLE under the model $\,\xbar{RM}_{\boldsymbol \delta}(\mathbf{A})$.
\item The MLE $\tilde{\boldsymbol \delta}_{\boldsymbol y}$ under $\,\xbar{RM}_{\boldsymbol \delta}(\mathbf{A})$ always exists and is the unique point of $\mathcal{X}_{\mathbf{A}}$ which satisfies:
\begin{align}
&\mathbf{A}\boldsymbol{\delta} = \gamma \mathbf{A} \boldsymbol q, \mbox{ for some } \gamma > 0; \label{A} \\
&\boldsymbol 1'\boldsymbol \delta = 1 \qquad (\mbox{only for } \boldsymbol \delta \equiv \boldsymbol p). \nonumber
\end{align}
\end{enumerate}
\end{theorem}

The vector $\tilde{\boldsymbol \delta}_{\boldsymbol y}$ is called the extended MLE of $\boldsymbol \delta$ under the relational model. The proof is given in the Appendix. The following example illustrates the theorem.

\vspace{5mm}

\noindent{\bf Example \ref{exConstruct}} (revisited):

Notice first that $F =\{1,2,4,5\}$ is a facial set of $\mathbf{A}$. The support of the observed distribution  $supp(\boldsymbol q) =\{1,2,4\}$ is a subset of $F$. 
Therefore, the MLE of $\boldsymbol q$ exists in the closure of the relational model. As it was shown earlier, the distribution $\boldsymbol p_0 =  (1/8, 1/2, 0, 1/4, 1/8)'$ is in $\xbar{RM}_{\boldsymbol p}(\mathbf{A})$. As $\mathbf{A}\boldsymbol p_0 =  7/8 \mathbf{A}\boldsymbol q$, the extended MLE of $\boldsymbol q$ is $\boldsymbol p_0$.
\qed

\vspace{3mm}

The next theorem establishes a condition under which the maximum likelihood estimates of the model parameters under an extended relational model exist:

\begin{theorem}\label{MLEmodelParam}
Assume that the MLE $\hat{\boldsymbol \delta}$ under the extended relational model $\, \xbar{RM}_{\boldsymbol \delta}(\mathbf{A})$ exists.  The  maximum likelihood estimates of the model parameters $\boldsymbol \theta$ exist if and only if  $supp(\hat{\boldsymbol \delta})$ is $\mathbf{A}$-feasible.
\end{theorem}

\begin{proof}

By Theorem \ref{Thfactors}, the distribution $\hat{\boldsymbol \delta}$ factors according to $\mathbf{A}$ if and only if $supp(\hat{\boldsymbol \delta})$ is $\mathbf{A}$-feasible. In this case  $\hat{\delta}(i) = \prod_{j=1}^J \hat{\theta}_j^{a_{ij}}$ for all $i \in \mathcal{I}$, and, by uniqueness, $\hat{\boldsymbol \theta}=(\hat{\theta}_1, \dots, \hat{\theta}_J)'$ are the maximum likelihood estimates of the model parameters.
\end{proof}

If $supp(\hat{\boldsymbol \delta})$ is not $\mathbf{A}$-feasible, then  $\hat{\boldsymbol \delta}$ is the limit of a sequence of the positive distributions in the model which factor according to $\mathbf{A}$. Although the cell parameters of these distributions can be factored using some model parameters $\boldsymbol \theta^{(n)} > \boldsymbol 0$, the limits of individual components of  $\boldsymbol \theta^{(n)}$, as $n \to 0$, may not exist. In the case of the log-linear models this fact was illustrated by \cite{RinaldoTR}. The same situation occurs in the construction of Example \ref{exConstruct}, where $\theta_2^{(n)} \to \infty$ as $n \to \infty$.

As Theorem  \ref{MLEextendTHnew} implies, the MLE in the extended relational model can be obtained using the MLE in a non-extended model.  \cite{KRipf1} proposed a generalized iterative scaling procedure, called G-IPF, for computing the MLE under (non-extended) relational models. The algorithm relies on the condition that $\mathbf{A}\boldsymbol q > \boldsymbol 0$. Every iteration of this procedure implements the following algorithm, IPF($\gamma$), for a specific value of $\gamma$.

\vspace{1mm}

\begin{center} \textbf{IPF($\gamma$) Algorithm:} \end{center}

\noindent {Set} $n = 0$; \,\, ${\delta}_{\gamma}^{(0)}(i) = 1$ for all $i \in \mathcal{I}$, \,\, and proceed as follows.
\begin{enumerate}
\item[] {\tt Step 1}: {Find} $j \in \{1,2,\dots,J\}$, such that $n+1 \equiv j \mbox{ mod } J$;
\item[] {\tt Step 2}: {Compute}  
\begin{eqnarray}  
\delta_{\gamma}^{(n+1)}(i) &=& \delta_{\gamma}^{(n)}(i) \left(\gamma \frac{A_{j}\boldsymbol{q}}{A_{j}\boldsymbol{\delta}_{\gamma}^{(n)}}\right)^{a_{ji}} \,\, \mbox{for all } i \in \mathcal{I}. \label{RipfGamma} 
\end{eqnarray}
\item[] {\tt Step 3}: While  $\gamma A_{j}\boldsymbol{q} \neq A_{j}\boldsymbol{\delta}_{\gamma}^{(n+1)}$ for at least one $j$,  set $n = n+1$, go to {\tt Step 1}.
\item[] {\tt Step 4}: Set $\boldsymbol{\delta}_{\gamma}^{*}=\boldsymbol{\delta}_{\gamma}^{(n)}$, and finish. \qed
\end{enumerate}

\vspace{3mm}

\vspace{2mm}
The G-IPF algorithm commences with executing IPF($\gamma$) for $\gamma = 1$, which is sufficient to compute the MLE in the case of probabilities with the overall effect and in the case of intensities. If in the case of probabilities the overall effect is not present, G-IPF updates $\gamma$  and calls IPF($\gamma$) again. The procedure is repeated until, for some $\gamma$, the limit vector $\boldsymbol{\delta}_{\gamma}^{*}$ sums to $1$, and thus is a parameter of a non-negative probability distribution. The variant of G-IPF, which employs the bisection method to update $\gamma$, is described in the following.

\begin{center} \textbf{G-IPF Algorithm:} \end{center}


\begin{itemize}
\item[] If $\boldsymbol \delta \equiv \boldsymbol \lambda$,  compute $\tilde{\boldsymbol{\lambda}}$ using IPF(1), and finish.
\item[] If $\boldsymbol \delta \equiv \boldsymbol p$, compute $\boldsymbol{p}^*$ using IPF($1$). \\
If $\boldsymbol 1\boldsymbol{p}^* = 1$, set $\tilde{\boldsymbol p} = \boldsymbol p^*$, and finish. 
Otherwise, \\
{compute} $\gamma_L = (\boldsymbol 1'\mathbf{A}\boldsymbol q)^{-1}$, $\gamma_R = \mbox{min } \{1/A_1\boldsymbol q, \dots, 1/A_J\boldsymbol q\}$, and proceed as follows:
\begin{itemize}
\item[] {\tt Step 1}: {Find} $\boldsymbol \delta_{(\gamma_L +\gamma_R)/2}^*$ \hspace{2mm} using IPF($\gamma$). \\
\item[] {\tt Step 2}: While   $\boldsymbol 1\boldsymbol \delta_{(\gamma_L +\gamma_R)/2}^* \ne 1$, 
\begin{itemize}
\item[] \qquad {if }  $\boldsymbol 1\boldsymbol{\delta}_{(\gamma_{L} + \gamma_{R})/2}^* < 1$,  {set } $\gamma_{L}  = \frac{\gamma_{L} + \gamma_{R}}{2}$,  
\item[] \qquad {else } {set } $\gamma_{R}  = \frac{\gamma_{L} + \gamma_{R}}{2}$;
\item[] \qquad {go to {\tt Step 1}}.\\
\end{itemize} 
\item[] {\tt Step 3}: Set  $\tilde{\boldsymbol{p}}= \boldsymbol{\delta}_{(\gamma_L +\gamma_R)/2}^*$, and finish. \qed
\end{itemize}
\end{itemize}

\vspace{3mm}

If $\mathbf{A}\boldsymbol q > \boldsymbol 0$, the G-IPF algorithm applies to the extended case directly. 

\begin{theorem}\label{ThGammaAll}
Let $\boldsymbol y$ be the vector of observed frequencies under Poisson or multinomial sampling, with $\boldsymbol q$ defined in (\ref{b}), and let $RM_{\boldsymbol \delta}(\mathbf{A})$ be a relational model. Assume that $\mathbf{A} \boldsymbol q > \boldsymbol 0$. The G-IPF algorithm converges to the MLE $\tilde{\boldsymbol \delta}_{\boldsymbol y}$  under $\,\xbar{RM}_{\boldsymbol \delta}(\mathbf{A})$.
\end{theorem}

\begin{proof}
As $\mathbf{A} \boldsymbol q > \boldsymbol 0$, the IPF-sequence $\boldsymbol{\delta}_{\gamma}^{(n)}$ defined in (\ref{RipfGamma}) is positive, and the proof of its convergence in \citet[Theorem 3.2]{KRipf1} applies. In particular, the limit of the sequence, $\boldsymbol{\delta}_{\gamma}^{*}$, satisfies $\mathbf{A}\boldsymbol{\delta}_{\gamma}^* = \gamma \mathbf{A} \boldsymbol q$, and, for an arbitrary kernel basis matrix $\mathbf{D}$, $\mathbf{D} \mbox{log } \boldsymbol{\delta}_{\gamma}^{(n)} = \boldsymbol 0$  for all $n \in \mathbb{Z}_{\geq 0}$. The latter implies that $\, \boldsymbol{\delta}_{\gamma}^{(n)} \in \mathcal{X}_{\mathbf{A}}$ for all $n$, and, as $\mathcal{X}_{\mathbf{A}}$ is a closed set in $\mathbb{R}^{|\mathcal{I}|}_{\geq 0}$, $\boldsymbol{\delta}_{\gamma}^{*} \in \mathcal{X}_{\mathbf{A}}$.

Let $\boldsymbol{\delta}_{1}^{*}$ be the limit vector obtained from IPF($1$), and thus $\boldsymbol{\delta}_{1}^{*} \in \mathcal{X}_{\mathbf{A}}$ and $\mathbf{A} \boldsymbol{\delta}_{1}^* = \mathbf{A} \boldsymbol q$. 

Suppose $\boldsymbol \delta \equiv \boldsymbol \lambda$.  Then,  as (\ref{A}) holds  for $\boldsymbol{\delta}_{1}^{*}$ with $\gamma = 1$, Theorem \ref{MLEextendTHnew}({\it{iii}}) implies that  $\boldsymbol{\delta}_{1}^{*}$ is equal to the extended MLE: $\tilde{\boldsymbol \delta}_{\boldsymbol y} = \boldsymbol{\delta}_{1}^{*}$.

Suppose $\boldsymbol \delta \equiv \boldsymbol p$. First, assume that the overall effect is present, and thus  there exists a $\boldsymbol k \in \mathbb{R}^J_{\geq 0}$, such that $\boldsymbol 1' = \boldsymbol k' \mathbf{A}$. The latter yields that $\boldsymbol 1' \boldsymbol \delta _1^* = \boldsymbol k' \mathbf{A} \boldsymbol \delta _1^* =   \boldsymbol k' \mathbf{A} \boldsymbol q = \boldsymbol 1' \boldsymbol q = 1$. Therefore, (\ref{A}) holds  for $\boldsymbol{\delta}_{1}^{*}$ with $\gamma = 1$. By Theorem \ref{MLEextendTHnew}({\it{iii}}), $\tilde{\boldsymbol \delta}_{\boldsymbol y} = \boldsymbol{\delta}_{1}^{*}$. 

Now, assume that the overall effect is not present. In this situation, G-IPF updates $\gamma$  and calls IPF($\gamma$); and this procedure is repeated until a $\gamma^*$ for which the IPF-limit $\boldsymbol{\delta}_{\gamma^*}^{*}$ sums to $1$ is found. Then, $\boldsymbol{\delta}_{\gamma^*}^{*}$ satisfies (\ref{A}) with $\gamma = \gamma^*$. By Theorem \ref{MLEextendTHnew}({\it{iii}}), $\tilde{\boldsymbol \delta}_{\boldsymbol y} = \boldsymbol{\delta}_{1}^{*}$. 
\end{proof}

Next, it is shown how G-IPF can be used if the condition $\mathbf{A}\boldsymbol q > \boldsymbol 0$ does not hold. Let $\mathcal{J}_{0} = \{j \in \{1, \dots, J\}: \,\,  A_j \boldsymbol q = 0\}$, and assume that $\mathcal{J}_{0} \neq \emptyset$. Further, let $\mathcal{I}_{0} = \{i \in \mathcal{I}: \,\, \exists j \in \mathcal{J}_0 \quad a_{ji} = 1\}$, and let  $\mathcal{I}_{*} = \mathcal{I} \setminus \mathcal{I}_{0}$. Denote by $\mathbf{A}_*$ the matrix obtained from $\mathbf{A}$ by removing the columns with indices in $\mathcal{I}_0$ and by removing the zero rows, if such occur afterwards, and by  $\boldsymbol \delta_*$, $\boldsymbol y_*$, and $\boldsymbol q_*$ the corresponding sub-vectors of $\boldsymbol \delta$, $\boldsymbol y$, and $\boldsymbol q$. By Theorem \ref{MLEextendTHnew}({\it{iii}}), the MLE $\tilde{\boldsymbol \delta}_{\boldsymbol y_*}$ of $\boldsymbol y_*$ under $\,\xbar{RM}_{\boldsymbol \delta_*}(\mathbf{A}_*)$ exists and is unique. Since $\mathbf{A}_*\boldsymbol q_* > \boldsymbol 0$, $\tilde{\boldsymbol \delta}_{\boldsymbol y_ *}$ can be computed using G-IPF, see Theorem \ref{ThGammaAll}, and the following holds:

\begin{theorem}\label{GIPFzeros}
The MLE of $\boldsymbol y$ under $\,\xbar{RM}_{\boldsymbol \delta}(\mathbf{A})$ is equal to $\tilde{\boldsymbol \delta}_{\boldsymbol y} = (\tilde{\boldsymbol \delta}_{\boldsymbol y_*}, \boldsymbol 0_{\mathcal{I}_0})$.
\end{theorem}

\begin{proof}

In order to show that $\tilde{\boldsymbol \delta}_{\boldsymbol y} \in \mathcal{X}_{\mathbf{A}}$, it will first be verified that  $\mathcal{I}_*$ is a facial set of $\mathbf{A}$. Let $\boldsymbol a_i$ be the $i$-th column of $\mathbf{A}$, then, with  $\boldsymbol c = (\boldsymbol 0_{\mathcal{J}\setminus \mathcal{J}_0}, \boldsymbol 1_{\mathcal{J}_0})'$, $\boldsymbol c ' \boldsymbol a_{i} = 0$ for any $i \in \mathcal{I}_*$.  If $i \notin \mathcal{I}_*$, then $a_{ji} = 1$ for some $j \in \mathcal{J}_0$, and thus  $\boldsymbol c ' \boldsymbol a_{i} > 0$. Therefore, $\mathcal{I}_{*}$ is a facial set of $\mathbf{A}$. Then, by Lemma \ref{FaceVariety}, $\tilde{\boldsymbol \delta}_{\boldsymbol y} \in \mathcal{X}_{\mathbf{A}}$.

Next, in the case of probabilities, the normalization condition $\boldsymbol 1'_{\mathcal{I}_*}\tilde{\boldsymbol \delta}_{\boldsymbol y_*} = 1$ implies that   $\boldsymbol 1'\tilde{\boldsymbol \delta}_{\boldsymbol y}  = 1$. Further, $\mathbf{A}_*\tilde{\boldsymbol \delta}_{\boldsymbol y_*} = \gamma \mathbf{A}_* \boldsymbol q_{*}$ implies that $\mathbf{A}\tilde{\boldsymbol \delta}_{\boldsymbol y} = \gamma \mathbf{A} \boldsymbol q$. 

Finally, by Theorem \ref{MLEextendTHnew}({\it{iii}}), $\tilde{\boldsymbol \delta}_{\boldsymbol y}$ is the MLE of $\boldsymbol y$ under $\,\xbar{RM}_{\boldsymbol \delta}(\mathbf{A})$. 
\end{proof}

\vspace{1mm}

\section{Conclusion}

Some research areas deal with populations of a complex structure to which inference based on the standard log-linear approach does not apply, but the relational model framework can be used. The relational models are more flexible as they allow effects associated with arbitrary subsets of cells, can be used for incomplete tables, and do not require the presence of an overall effect.  Similarly to the log-linear case, data with zero counts may not possess an MLE under a relational model. A necessary and sufficient condition for the existence of the MLE was obtained in Section  \ref{RMmlePositive}. When this condition does not hold, an MLE may exist in the extended sense, that is, in the closure of the relational model.  Different but equivalent ways of defining such a closure, and a necessary and sufficient condition for the existence of the extended MLE in it were presented in Section \ref{SecClosure}. A condition under which a distribution in the closure factorizes according to the model matrix was also given. These results were obtained using concepts and methods of algebraic statistics. Just like in the case of relational models, the cases of multinomial and Poisson sampling are not equivalent. It was shown in Section \ref{MLEsection}, that the generalized relative proportional fitting procedure originally suggested for relational models also works when the data contain zeros and the MLE is sought for in the closure of a relational model.

\appendix

\section{Appendix}

\subsection{Properties of facial sets}

\begin{lemma} \label{LemmaFacial}
Let $\mathbf{A}$ be the model matrix of a relational model, and let $F$ be a facial set of $\mathbf{A}$. Then:
\begin{enumerate}[(i)]
\item \label{a} There exists a $\boldsymbol c \in \mathbb{R}^J$, such that $\boldsymbol c'\boldsymbol a_i = 0$ for any $i \in F$ and $\boldsymbol c'\boldsymbol a_i  > 0$ for any $i \notin F$.  
\item \label{aa} For any $\boldsymbol d \in Ker(\mathbf{A})$, either both $supp(\boldsymbol d^+) \subseteq F$ and $supp(\boldsymbol d^-) \subseteq F$ or both $supp(\boldsymbol d^+) \nsubseteq F$ and $supp(\boldsymbol d^-) \nsubseteq F$.
\item \label{aaa} For any $\boldsymbol \delta \in \mathcal{X}_{\mathbf A}$,  either $supp(\boldsymbol \delta) = \mathcal{I}$ or $supp(\boldsymbol \delta)$ is a facial set of $\mathbf{A}$.
\item \label{aaaa} If $F$ is a facial set of $\mathbf{A}$, there exists a $\boldsymbol \delta \in \mathcal{X}_{\mathbf{A}}$, such that $supp(\boldsymbol \delta) = F$.
\end{enumerate}
\end{lemma}

The statements of the lemma were proved by \cite{GeigerMeekSturm2006} and \cite*{RauhMatroid} for models of type (\ref{genll}) when the overall effect is present. Their proofs do not rely on the latter characteristic and thus apply here.

The next lemma shows that the condition of existence of the MLE given in Theorem \ref{MLEnoOvEff} can also be formulated in terms of facial sets.

\begin{lemma}\label{LemmaFacialTau}
There exists a $\boldsymbol z > \boldsymbol 0$, such that $\mathbf{A}\boldsymbol z = \mathbf{A} \boldsymbol q$, if and only if
 $supp(\boldsymbol q)$ is not contained in any facial set of $\mathbf{A}$.
\end{lemma}

\begin{proof}

Suppose there exists a $\boldsymbol z > \boldsymbol 0$, such that $\mathbf{A}\boldsymbol z = \mathbf{A} \boldsymbol q$, and thus  $\boldsymbol d = \boldsymbol z - \boldsymbol q \in Ker(\mathbf{A})$ and $\boldsymbol q + \boldsymbol d > \boldsymbol 0$. 

Let $F$ be a facial set of $\mathbf{A}$. If both $\boldsymbol d^+ \subseteq F$ and $\boldsymbol d^- \subseteq F$, then  $d_i = 0$ for all $i \notin F$. Because $\boldsymbol q + \boldsymbol d  > \boldsymbol 0$, $q_i + d_i = q_i > 0$ for all $i \notin F$. Therefore, $supp(\boldsymbol q)$ is not contained in $F$.  Otherwise, see Lemma \ref{LemmaFacial}, both $\boldsymbol d^+ \nsubseteq F$ and $\boldsymbol d^- \nsubseteq F$, and there exists an $i \notin F$ such that $d_i < 0$. If $q_i$ was zero, then $q_i + d_i$ would be negative, which contradicts the initial assumption $\boldsymbol q + \boldsymbol d > \boldsymbol 0$. Therefore, $q_i$ has to be positive, which implies that $supp(\boldsymbol q)$ is not contained in $F$.

To prove the converse, assume that $supp(\boldsymbol q)$ is not contained in any facial set $F$. Suppose the equation $\mathbf{A} \boldsymbol q = \mathbf{A}\boldsymbol z$ has no (strictly) positive solution in $\boldsymbol z$, and, therefore, $\mathbf{A}\boldsymbol q \notin relint(C_{\mathbf{A}})$. A non-negative solution always exists, and thus  $\mathbf{A}\boldsymbol q$ belongs to a face of $C_{\mathbf{A}}$.  Then (\ref{linComb}) holds for $\boldsymbol t =\mathbf{A} \boldsymbol q$ for some facial set $F$; without loss of generality, $F=\{1, \dots,f\}$:
\begin{equation*}
\mathbf{A}\boldsymbol q = s_1\boldsymbol a_1 + \dots +s_f\boldsymbol a_f.
\end{equation*}
Hence,
\begin{equation}\label{eqFn}
(q_1- s_1)\boldsymbol a_1 + \dots + (q_f-s_f)\boldsymbol a_f + q_{f+1}\boldsymbol a_{f+1} + \dots +q_{|\mathcal{I}|}\boldsymbol a_{|\mathcal{I}|} = \boldsymbol 0.
\end{equation}
Multiplying both sides of (\ref{eqFn}) by a vector $\boldsymbol c$, such that $\boldsymbol c'\boldsymbol a_i = 0$ for $i \in F$ and $\boldsymbol c'\boldsymbol a_i > 0$ for $i \notin F$, leads to:
$$q_{f+1} = 0, \dots, q_{|\mathcal{I}|} = 0,$$
which means that $supp(\boldsymbol q) \subset F$. This contradicts the initial assumption that $supp(\boldsymbol q)$ is not contained in any facial set. 
\end{proof}

The following lemma is used in the proofs of Theorems \ref{MLEextendTHnew} and \ref{GIPFzeros}. 

\begin{lemma}\label{FaceVariety}
If $F$ is a facial set of $\mathbf{A}$, then, for any $\boldsymbol \delta_F \in \mathcal{X}_{\mathbf{A}_F}$, $\boldsymbol \delta = (\boldsymbol \delta_F, \boldsymbol 0_{\mathcal{I}\setminus F} )\in \mathcal{X}_{\mathbf{A}}$.
\end{lemma}

\begin{proof} 
Take an arbitrary $\boldsymbol d \in Ker(\mathbf{A})$. As $F$ is a facial set of $\mathbf{A}$, by Lemma \ref{LemmaFacial}(\ref{aa}),  exactly one of the following holds:
$$supp(\boldsymbol d^+) \subseteq F \,\, \mbox{ and }\, supp(\boldsymbol d^-) \subseteq F, \quad \mbox{or} \quad supp(\boldsymbol d^+) \nsubseteq F \,\, \mbox{ and }\, supp(\boldsymbol d^-) \nsubseteq F.$$
In the first case, there exists a $\boldsymbol d_F \in Ker(\mathbf{A}_F)$, such that $\boldsymbol d = (\boldsymbol d_{F}, \boldsymbol 0_{\mathcal{I} \setminus F})$.   Since $\boldsymbol \delta_F \in \mathcal{X}_{\mathbf{A}_F}$, $(\boldsymbol \delta_F)^{\boldsymbol d_F^+}  = (\boldsymbol \delta_F)^{\boldsymbol d_F^-}$, and, therefore, 
$$(\boldsymbol \delta)^{\boldsymbol d^+} = (\boldsymbol \delta_F)^{\boldsymbol d_F^+}\cdot (\boldsymbol 0_{\mathcal{I}\setminus F})^{\boldsymbol 0_{\mathcal{I}\setminus F}} = (\boldsymbol \delta_F)^{\boldsymbol d_F^-} \cdot (\boldsymbol 0_{\mathcal{I}\setminus F})^{\boldsymbol 0_{\mathcal{I}\setminus F}} = (\boldsymbol \delta)^{\boldsymbol d^-}.$$
In the second case, there exist such $i_1, i_2 \notin F$ that $d_{i_1} > 0$ and $d_{i_2} < 0$, and thus,
$$(\boldsymbol \delta)^{\boldsymbol d^+} = (\boldsymbol \delta_F)^{\boldsymbol d_F^+}\cdot 0 =  (\boldsymbol \delta_F)^{\boldsymbol d_F^-}\cdot 0  = (\boldsymbol \delta)^{\boldsymbol d^-}.$$
As $(\boldsymbol \delta)^{\boldsymbol d^+} = (\boldsymbol \delta)^{\boldsymbol d^-}$ for any $\boldsymbol d \in Ker(\mathbf{A})$, $\boldsymbol \delta \in \mathcal{X}_{\mathbf{A}}$.
\end{proof}

\subsection{Proof of Theorem \ref{THclosure}}

The proof extends the arguments given by \cite{GeigerMeekSturm2006} and \cite{RauhMatroid}. It will be shown first that for any distribution in $\,\xbar{RM}_{\delta}(\mathbf{A})$ there exists a sequence of distributions in $RM_{\boldsymbol \delta}(\mathbf{A})$ that converges to it pointwise. 

Let $\boldsymbol \delta^* \in \,\xbar{RM}_{\delta}(\mathbf{A})$. By Lemma \ref{LemmaFacial}, as $\boldsymbol \delta^* \in \mathcal{X}_{\mathbf{A}}$, $F = supp(\boldsymbol \delta^*)$ is either $\mathcal{I}$ or a facial set of $\mathbf{A}$. If $F =\mathcal{I}$, then $\boldsymbol \delta^* > \boldsymbol 0$, and the statement holds with $\boldsymbol \delta^{(n)} \equiv \boldsymbol \delta^*$.  Assume that $F \subsetneq \mathcal{I}$. For simplicity of exposition, let $F = \{1,\dots,f\}$, and then $\boldsymbol \delta^* = (\delta_1^*, \dots, \delta_f^*, 0, \dots, 0)$.

First, find $\eta_1, \dots, \eta_J > 0$ that satisfy:
$$\prod_{j=1}^J \eta_j^{a_{ji}} = \delta_i^* \quad \mbox{ for } i \in F.$$
The existence of such $\theta$'s can be proved using the same argument as \citet[p.28]{GeigerMeekSturm2006} gave for the case of extended log-linear models. By Lemma \ref{LemmaFacial}, there exists a $\boldsymbol c = (c_1, \dots, c_J)' \in \mathbb{R}^J$, such that $\boldsymbol c'\boldsymbol a_i = 0$ for all $i \in F$  and $\boldsymbol c'\boldsymbol a_i > 0$  for any $i \notin F$. 
Order the columns of $\mathbf{A}$ so that $c_1 > 0$, and then order the rows of $\mathbf{A}$ so that $a_{11} = 1$. 

If $\boldsymbol \delta \equiv  \boldsymbol \lambda$, set, for $n \in \mathbb{Z}_{>0}$,
$$\lambda_i^{(n)} = \prod_{j=1}^J(n^{-c_j}\eta_j)^{a_{ji}},   \qquad i \in \mathcal{I}.$$
The distribution $\boldsymbol \lambda^{(n)}=  (\lambda_1^{(n)}, \dots, \lambda_{|\mathcal{I}|}^{(n)})'$ is positive and satisfies  (\ref{RMexpF}) with $\theta_j = n^{-c_j}\eta_j$. Therefore, $\boldsymbol \lambda^{(n)} \in RM_{\boldsymbol \lambda}(\mathbf{A})$. Further,   \begin{equation*}
\lim_{n \to \infty} \lambda_i^{(n)} =  \lim_{n \to \infty} n^{-\boldsymbol c' \boldsymbol a_i} \prod_{j=1}^J \eta_j^{a_{ji}} =\left\{ \begin{array}{ll} \delta_i^*,  & \mbox{ if } i \in F,\\
{}  \\
0, & \mbox{ if } i \notin F,\end{array}\right.  
\end{equation*}
thus $\boldsymbol \lambda^{(n)} \to \boldsymbol \delta^*$ pointwise, as $n \to \infty$.

If $\boldsymbol \delta \equiv \boldsymbol p$, take
$$
\eta_1^{(n)} = \frac{1-\sum_{i: \, a_{1i} = 0}  \prod_{j = 2}^J ( n^{-c_j}\eta_j)^{a_{ji}} }{\sum_{i:\, a_{1i} = 1}   \prod_{j = 2}^J ( n^{-c_j}\eta_j)^{a_{ji}} },
$$
and set 
$$p_i^{(n)} = (\eta_1^{(n)})^{a_{1i}} \prod_{j = 2}^J ( n^{-c_j}\eta_j)^{a_{ji}}, \qquad i \in \mathcal{I}.$$
The choice of $\eta_1^{(n)}$ implies that $\boldsymbol 1' \boldsymbol p^{(n)} =1$. As  $\boldsymbol p^{(n)} = (p_1^{(n)}, \dots, p_{|\mathcal{I}|}^{(n)})' $ is positive and satisfies  (\ref{RMexpF}) with $\theta_1 = \eta^{(n)}_1$, $\theta_j =  n^{-c_j}\eta_j$, for $j = 2, \dots, J$,  $\boldsymbol p^{(n)} \in RM_{\boldsymbol p}(\mathbf{A})$. Next, because $\boldsymbol c'\boldsymbol a_i = 0$ if $i \in F$,
\begin{eqnarray} \label{etaL}
\lim _{n \to \infty}n^{c_1}\eta_1^{(n)} &=&\lim _{n \to \infty}\frac{n^{c_1}(1-\sum_{a_{1i} = 0, i \in F} \prod_{j = 2}^J \eta_j^{a_{ji}} - \sum_{a_{1i} = 0, i \notin F} n^{- \boldsymbol c' \boldsymbol a_i} \prod_{j = 2}^J \eta_j^{a_{ji}} )}{n^{c_1}(\sum_{a_{1i} = 1, i \in F} \prod_{j = 2}^J \eta_j^{a_{ji}} + \sum_{a_{1i} = 1, i \notin F}  n^{-\boldsymbol c' \boldsymbol a_i} \prod_{j = 2}^J \eta_j^{a_{ji}} )}  \nonumber \\
&& \nonumber \\
&& \nonumber \\
&=&  \frac{1-\sum_{i \in F: \, a_{1i} = 0}  \prod_{j = 2}^J \eta_j^{a_{ji}} }{\sum_{i \in F: \, a_{1i} = 1}  \prod_{j = 2}^J \eta_j^{a_{ji}} }= \eta_1.
\end{eqnarray}

\vspace{1mm}

\noindent Further, for $i \in \mathcal{I}$, using (\ref{etaL}),
\begin{eqnarray*}
\lim_{n \to \infty} p_i^{(n)}&=& \lim_{n \to \infty} n^{a_{1i}c_1-\boldsymbol c' \boldsymbol a_i} (\eta_1^{(n)})^{a_{1i}}  \prod_{j = 2}^J  \eta_j^{a_{ji}} = \lim_{n \to \infty} n^{-\boldsymbol c' \boldsymbol a_i} (n^{c_1}\eta_1^{(n)})^{a_{1i}}  \prod_{j = 2}^J  \eta_j^{a_{ji}} \\
&& \\
&=&
\lim_{n \to \infty} n^{-\boldsymbol c' \boldsymbol a_i} (\eta_1)^{a_{1i}}  \prod_{j = 2}^J  \eta_j^{a_{ji}} = \lim_{n \to \infty} n^{-\boldsymbol c' \boldsymbol a_i} \prod_{j = 1}^J  \eta_j^{a_{ji}}= \left\{
\begin{array}{ll}
\delta_i^* & i \in F,\\
{} & \\
0 & i \notin F.
\end{array}
\right.
\end{eqnarray*}
Hence, $\boldsymbol p^{(n)} \to \boldsymbol \delta^*$ pointwise, as $n \to \infty$.

Therefore, $ \,\xbar{RM}_{\delta}(\mathbf{A}) \subset cl(RM_{\boldsymbol \delta}(\mathbf{A}))$.

To prove the converse, choose a $\boldsymbol \delta^* \in \,cl(RM_{\boldsymbol \delta}(\mathbf{A}))$. Then, $\boldsymbol \delta^*$ is a pointwise limit of a sequence of distributions in $RM_{\boldsymbol \delta}(\mathbf{A})$, and $\boldsymbol \delta^*$ is the pointwise limit of a sequence in  $\mathcal{X}_{\mathbf{A}}$.
As $\mathcal{X}_{\mathbf{A}}$ is closed in the topology of pointwise convergence \citep[cf.][]{GeigerMeekSturm2006}, $\boldsymbol \delta^* \in \mathcal{X}_{\mathbf{A}}$. If $\boldsymbol \delta \equiv \boldsymbol p$, both  $\boldsymbol \delta^*$ and the sequence converging to it belong to the simplex $\Delta_{|\mathcal{I}|-1}$. Therefore, $\boldsymbol \delta^* \in \,\xbar{RM}_{\delta}(\mathbf{A})$, and the proof is complete.  \qed

\subsection{Proof of Theorem \ref{MLEextendTHnew}:}

\vspace{4mm}

The statement ({\it{i}}) follows from Theorem \ref{MLEnoOvEff} and Lemma \ref{LemmaFacialTau}. \qed 

\vspace{3mm}

\noindent In order to prove ({\it{ii}}), notice first that, the smallest facial set $F$ of $\mathbf{A}$ which contains $supp(\boldsymbol q)$ is uniquely defined. In this case,  $\mathbf{A}_F \boldsymbol q_F \in relint(\mathbf{C}_{\mathbf{A}_F})$, and, therefore, $supp(\boldsymbol q)$ is not contained in any facial set of $\mathbf{A}_F$. By part  ({\it{i}}) of this theorem, the MLE $\hat{\boldsymbol \delta}_{\boldsymbol y_ F}$ under $RM_{\boldsymbol \delta_F}(\mathbf{A}_F)$ exists.

Let  $\tilde{\boldsymbol \delta}_{\boldsymbol y} = (\hat{\boldsymbol \delta}_{\boldsymbol y_F}, \boldsymbol 0_{\mathcal{I} \setminus F})$. By Lemma \ref{FaceVariety}, $\tilde{\boldsymbol \delta}_{\boldsymbol y} \in \mathcal{X}_{\mathbf{A}}$. If $\boldsymbol \delta \equiv \boldsymbol p$, $\boldsymbol 1'\hat{\boldsymbol p}_{\boldsymbol y_F} = 1\,$, and thus $\tilde{\boldsymbol p}_{\boldsymbol y}$ satisfies the normalization condition $\,\boldsymbol 1'\tilde{\boldsymbol p}_{\boldsymbol y} = 1$.  It will be shown next that $\tilde{\boldsymbol \delta}_{\boldsymbol y}$ maximizes the full log-likelihood of $\boldsymbol y$.

Let $\boldsymbol \delta \equiv \boldsymbol \lambda$. The log-likelihood under the model $RM_{\boldsymbol \lambda_F}(\mathbf{A}_F)$ is equal to 
$${l}_F(\boldsymbol q_F, \boldsymbol \lambda_F) = \sum_{i \in F} q_{Fi} \mbox{log } \lambda_{Fi} - \sum_{i \in F} \lambda_{Fi},$$
and for any $\boldsymbol \lambda_F > 0$, $\,\,{l}_F(\boldsymbol q_F, \boldsymbol \lambda_F)\leq {l}_F(\boldsymbol q_F, \hat{\boldsymbol \lambda}_{\boldsymbol y_F}).$

Let ${\boldsymbol \lambda} = (\boldsymbol \lambda_F', \boldsymbol 0)'$, and let  $\boldsymbol \lambda^{(n)}$ be the sequence that was described in the proof of Theorem \ref{THclosure}. The full log-likelihood of the elements of this sequence is
\begin{eqnarray*}
{l}(\boldsymbol q, \boldsymbol \lambda^{(n)}) &=& \sum_{i\in \mathcal{I}} q_i \mbox{log }\lambda_i^{(n)}  - 
\sum_{i \in \mathcal{I}} \lambda_{i}^{(n)} = \sum_{i \in F} q_i  \mbox{log }\lambda_i^{(n)} - \sum_{i \in \mathcal{I}} \lambda_{i}^{(n)}\\
&=&  \sum_{i \in F} q_i \mbox{log } \{ n^{-\boldsymbol c\boldsymbol a_i} \prod_{j = 1}^J  \theta_j^{a_{ji}}\} - \sum_{i \in \mathcal I}  n^{-\boldsymbol c\boldsymbol a_i} \prod_{j = 1}^J  \theta_j^{a_{ji}} \\
&=&  \sum_{i \in F}q_i \mbox{log } \{\prod_{j = 1}^J  \theta_j^{a_{ji}}\} - \sum_{i \in F}  \prod_{j = 1}^J  \theta_j^{a_{ji}} -   \sum_{i \notin F} n^{-\boldsymbol c\boldsymbol a_i} \prod_{j = 1}^J  \theta_j^{a_{ji}}\\
&=& {l}_F(\boldsymbol q_F, \boldsymbol \lambda_F) -\sum_{i \notin F}  n^{-\boldsymbol c\boldsymbol a_i} \prod_{j = 1}^J  \theta_j^{a_{ji}}.
\end{eqnarray*}
Therefore, 
\begin{equation}\label{LLineqInt}
{l}(\boldsymbol q, \boldsymbol \lambda^{(n)})  \leq {l}_F  (\boldsymbol q_F, \boldsymbol \lambda_F) \leq {l}_F(\boldsymbol q_F, \hat{\boldsymbol \lambda}_{\boldsymbol y_F}).
\end{equation}

Let $\boldsymbol \delta \equiv \boldsymbol p$. The log-likelihood under the model $RM_{\boldsymbol p_F} (\mathbf{A}_F)$ is equal to
$${l}_F(\boldsymbol q_F, \boldsymbol p_F) = \sum_{i = 1}^f q_{Fi} \mbox{log } p_{Fi},$$
and for any $\boldsymbol p_F > 0$, such that $\boldsymbol 1'\boldsymbol p_F=1$,
$\,{l}_F(\boldsymbol q_F, \boldsymbol p_F)\leq {l}_F(\boldsymbol q_F, \hat{\boldsymbol p}_{\boldsymbol y_F})$.

Let ${\boldsymbol p} = (\boldsymbol p_F', \boldsymbol 0)'$, and let  $\boldsymbol p^{(n)}$ be the sequence that was described in the proof of Theorem \ref{THclosure}.
The full log-likelihood of the elements of this sequence is
\begin{eqnarray*}
{l}(\boldsymbol q, \boldsymbol p^{(n)}) &=& \sum_{i \in \mathcal{I}} q_i \mbox{log }p_i^{(n)}  = \sum_{i \in F} q_i  \mbox{log }p_i^{(n)} \\
&=&  \sum_{i\in F} q_i \mbox{log } \{ (\theta_1^{(n)})^{a_{1i}} \prod_{j = 2}^J ( n^{-c_j}\theta_j)^{a_{ji}}\} =  \sum_{i \in F} q_i \mbox{log } \{ (\theta_1^{(n)})^{a_{1i}} n^{a_{1i}c_1 - \boldsymbol c'\boldsymbol a_i} \prod_{j = 2}^J \theta_j^{a_{ji}}\} \\
&=& \sum_{i\in F:\, a_{i1}=1} q_i \mbox{log }  \theta_1^{(n)} n^{c_1} \prod_{j = 2}^J \theta_j^{a_{ji}}+
\sum_{i \in F:\, a_{i1}=0} q_i \mbox{log } \prod_{j = 2}^J \theta_j^{a_{ji}} \\
&=& \sum_{i \in F:\, a_{i1}=1} q_i \mbox{log }  \prod_{j = 1}^J \theta_j^{a_{ji}} +
\sum_{i \in F:\, a_{i1}=0} q_i \mbox{log } \prod_{j = 1}^J \theta_j^{a_{ji}} -  \sum_{i \in F:\, a_{i1}=1} q_i \mbox{log }\{\theta_1/ 
(\theta_1^{(n)} n^{c_1})\}\\
&=& {l}_F(\boldsymbol q_F, \boldsymbol p_F) - \mbox{log }\{\theta_1/ 
(\theta_1^{(n)} n^{c_1})\}\cdot \sum_{i \in F:\, a_{i1}=1} q_i.
\end{eqnarray*}

It will be shown next that $\theta_1/(\theta_1^{(n)} n^{c_1}) > 1$.

\begin{eqnarray*}
\frac{\theta_1}{\theta_1^{(n)} n^{c_1}} &=& \frac{1-\sum_{i \in F: \, a_{i1} = 0}  \prod_{j = 2}^J \theta_j^{a_{ji}} }{\sum_{i \in F: \, a_{i1} = 1} \prod_{j = 2}^J \theta_j^{a_{ji}} } \\
&\cdot& \frac{n^{c_1}(\sum_{a_{i1} = 1, i \in F}  \prod_{j = 2}^J \theta_j^{a_{ji}} + \sum_{a_{i1} = 1, i \notin F} n^{-\boldsymbol c'\boldsymbol a_i} \prod_{j = 2}^J \theta_j^{a_{ji}})}{n^{c_1}(1-\sum_{a_{i1} = 0, i \in F}  \prod_{j = 2}^J \theta_j^{a_{ji}} - \sum_{a_{i1} = 0, i \notin F} n^{-\boldsymbol c'\boldsymbol a_i} \prod_{j = 2}^J \theta_j^{a_{ji}})} \\
&=&\left(1+\frac{ \sum_{a_{i1} = 1, i \notin F} n^{-\boldsymbol c'\boldsymbol a_i} \prod_{j = 2}^J \theta_j^{a_{ji}}}{\sum_{i \in F: \, a_{i1} = 1}   \prod_{j = 2}^J \theta_j^{a_{ji}} }\right) / \left(1 - \frac{\sum_{a_{i1} = 0, i \notin F} n^{-\boldsymbol c'\boldsymbol a_i} \prod_{j = 2}^J \theta_j^{a_{ji}}}{\sum_{i \in F: \, a_{i1} = 0}  \prod_{j = 2}^J \theta_j^{a_{ji}} } \right) > 1.
\end{eqnarray*}
Therefore, 
\begin{equation}\label{LLineqPr}
{l}(\boldsymbol q, \boldsymbol p^{(n)})  \leq {l}_F(\boldsymbol q_F, \boldsymbol p_F) \leq {l}_F(\boldsymbol q_F, \hat{\boldsymbol p}_{\boldsymbol y_F}).
\end{equation}
Combining (\ref{LLineqInt}) and (\ref{LLineqPr}),
\begin{equation}\label{LLineqBoth}
{l}(\boldsymbol q, \boldsymbol \delta^{(n)})  \leq {l}_F(\boldsymbol q_F, \boldsymbol \delta_F) \leq {l}_F(\boldsymbol q_F, \hat{\boldsymbol \delta}_{\boldsymbol y_F}),
\end{equation}
and 
$$ \sup_n{l}(\boldsymbol q, {\boldsymbol \delta}^{(n)}) \leq {l}_F(\boldsymbol q_F, \hat{\boldsymbol \delta}_{\boldsymbol y_F}) .$$ 
Hence, whenever $\tilde{\boldsymbol \delta}^{(n)} \to \tilde{\boldsymbol \delta}$ as $n \to \infty$, 
${l}(\boldsymbol q, \tilde{\boldsymbol \delta}^{(n)})  \to {l}_F(\boldsymbol q_F, \hat{\boldsymbol \delta}_{\boldsymbol y_F}).$

Therefore, ${l}(\boldsymbol q, {\tilde{\boldsymbol \delta}_{\boldsymbol y}}) = \sup {l}(\boldsymbol q, {\boldsymbol \delta}) = {l}_F(\boldsymbol q_F, \hat{\boldsymbol \delta}_{\boldsymbol y_F})$, which concludes the proof of  ({\it{ii}}). 
\qed

\vspace{3mm}
\noindent The uniqueness claim in ({\it{iii}}) follows from the convexity of the log-likelihood function. The proof is similar to the one given by \citet[Proposition 4.7]{Lauritzen} for the case of extended log-affine models, and is thus omitted. In order to prove the second claim, suppose first that there exists a facial set $F$ such that $supp(\boldsymbol q) \subseteq F$. Let $F$ be the minimal of such sets. As shown in the proof of ({\it{ii}}), the MLE $\hat{\boldsymbol \delta}_{\boldsymbol y_F}$  under  ${RM}_{\boldsymbol \delta_F}(\mathbf{A}_F)$ exists, and, from (\ref{MLEsys}),
$$\mathbf{A}_F \hat{\boldsymbol{\delta}}_{\boldsymbol y_F} = \gamma \mathbf{A}_F \boldsymbol q_F, \mbox{ for some } \gamma > 0, \mbox{ and, if } \boldsymbol \delta \equiv \boldsymbol p, \,\, \boldsymbol 1'\hat{\boldsymbol \delta}_{\boldsymbol y_F} = 1.$$ 
The MLE under $\,\xbar{RM}_{\boldsymbol \delta}(\mathbf{A})$ is equal to 
 $\tilde{\boldsymbol \delta}_{\boldsymbol y} = (\hat{\boldsymbol \delta}_{\boldsymbol y_F}, \boldsymbol 0_{\mathcal{I}\setminus F})$. As $\tilde{\boldsymbol \delta}_{\boldsymbol y, F, i} = 0$ for $i \notin F$,
$\mathbf{A}\hat{\boldsymbol{\delta}}_{\boldsymbol y} = \gamma \mathbf{A} \boldsymbol q$, and, in the case of probabilities, $\boldsymbol 1'\hat{\boldsymbol \delta}_{\boldsymbol y_F} = 1$.

If, for all facial sets $F$, $supp(\boldsymbol q) \not \subseteq F$, then the MLE $\tilde{\boldsymbol \delta}_{\boldsymbol y}$ under the extended model exists and is also the MLE under $RM_{\boldsymbol \delta}(\mathbf{A})$. In this case,  (\ref{MLEsys}) holds and is the same as (\ref{A}), which completes the proof.
\qed


\section*{Acknowledgments}

The second author was supported in part by Grant K-106154 from the Hungarian National Scientific Research Fund (OTKA).

\bibliographystyle{plainnat}
\bibliography{uwthesis1021}

\end{document}